\documentclass[final,3p]{elsarticle}
\usepackage{amsthm}
\newtheorem{remark}{Remark}
\usepackage{booktabs}

 \usepackage{graphics}
 \usepackage{graphicx}
\usepackage{caption}
\usepackage{float}
\usepackage{subcaption}
\usepackage{amsmath}
\usepackage{color}
\usepackage{soul}           
\usepackage{ulem}
\usepackage{tikz}
\usepackage{makecell}
\usepackage{wrapfig}
\usepackage{algorithm}

\usepackage{algorithm}
\usepackage{algorithmic}

\usetikzlibrary{shapes.geometric, arrows}

\tikzstyle{normal} = [rectangle, rounded corners, minimum width=3cm, minimum height=1cm,text centered, draw=black, fill=red!30,text width=10em]

\tikzstyle{robust} = [trapezium, trapezium left angle=70, trapezium right angle=110, minimum width=3cm, minimum height=1cm, text centered, draw=black, fill=blue!30,text width=4em]

\newtheorem{proposition}{Proposition}[section]

\usepackage[T1]{fontenc}
\usepackage{hyperref}
\hypersetup{
colorlinks=true
}
\usepackage{epstopdf}
\usepackage{xcolor}

\biboptions{sort&compress}

\journal{Measurement}

\begin{document}
\begin{frontmatter}

 \title{From Noise to Prognosis: A Physics-Grounded, Fractional-Domain Framework for Early Gear Fault Detection in Aviation Drivetrains}

\author{Yaakoub Berrouche}\corref{cor1}\ead{berrouchey@gmail.com}


\cortext[cor1]{Corresponding author.}

\address{Faculty of technology, Department of Electronics, Ferhat Abbas University Setif 1, Setif, Algeria }

\begin{abstract}

Early and reliable detection of gear faults in complex drivetrain systems is critical for aviation safety and operational availability. We present the Local Damage Mode Extractor (LDME), a structured, physics-informed signal processing framework that combines dual-path denoising, multiscale decomposition, fractional-domain enhancement, and statistically principled anomaly scoring to produce interpretable condition indicators without supervision. LDME is organized in three layers: (i) dual-path denoising (DWT with adaptive Savitzky–Golay smoothing) to suppress broadband noise while preserving transient fault structure; (ii) multi-scale damage enhancement using a Teager–Kaiser pre-amplifier followed by a Hadamard–Caputo fractional operator that accentuates non-sinusoidal, low-frequency fault signatures; and (iii) decision fusion where harmonics-aware Fourier indicators are combined and scored by an unsupervised anomaly detector. Evaluation using the Case Western Reserve University (CWRU) bearing dataset, the HUMS 2023 planetary gearbox benchmark, and a controlled simulated dataset shows that LDME consistently distinguishes nominal, early-crack, and propagated-crack stages under various operating situations. LDME identifies the primary detection event earlier (198 cycles) than HT-TSA (284 cycles) and advances maintenance recommendation time from 383 to 365 cycles. We discuss its relation to prior art, limitations, and future theoretical directions. All code and experimental configurations are documented for reproducibility.

\end{abstract}

\begin{keyword}
gear fault detection \sep vibration analysis \sep anomaly detection \sep cyclostationarity \sep fractional filtering \sep HUMS\sep CWRU.
\end{keyword}

\end{frontmatter}

\section{Introduction}
Rotating machinery health monitoring is a mature but still open problem when early-stage, low-amplitude damage must be detected in the presence of strong modulation and environmental noise. Planetary gearboxes are particularly challenging: multi-body kinematics, gear-mesh modulation and closely spaced harmonics often mask incipient crack signatures \cite{he2026adaptive, shang2026enhanced}. Time Synchronous Averaging (TSA) and frequency-domain spectral analysis remain canonical tools but degrade under low SNR and non-stationarity \cite{braun1975extraction,mcfadden2000application}. Recent work has explored physics-informed signal processing and data-driven methods; see for example surveys on vibration-based condition monitoring and rolling-element diagnostics \cite{lei2014condition, randall2011rolling}.

In order to assist structural usage monitoring, modern Health and Usage Monitoring Systems (HUMS) continuously gather a variety of aircraft parameters, such as attitude, altitude, airspeed, torque, and others. Instead than depending on conservative worst-case assumptions, usage monitoring allows for the assignment of real-time consumption or damage from individual flights by quantifying the actual operational load faced by each aircraft component \cite{bechhoefer2024hunting}. The operational lifespan of important components may be increased by more effectively utilizing their genuine service life through precise component utilization measurement. Precise identification of flight regimes, which are the aircraft's immediate operational state or profile during a flight, is necessary to accomplish this.Every component exposed to operational load is given a suitable damage factor for each recognized regime, which serves as a mathematical foundation for usage-based maintenance and life evaluation.

The Case Western Reserve University (CWRU) Bearing Dataset \cite{SMITH2015100} is widely employed to benchmark novel network architectures and signal processing techniques for rolling element bearing fault diagnosis. A comprehensive discussion of its use as a benchmarking tool is provided by Smith and Randall \cite{SMITH2015100}. The dataset comprises vibration signals acquired from a 2~HP electric motor, in which either the drive-end or fan-end bearing has been replaced with artificially seeded faults. Experiments are conducted under varying motor loads from 0~HP to 3~HP in 1~HP increments, and at rotational speeds ranging from 1730~RPM to 1797~RPM. For each bearing location and operating condition, three fault sizes are introduced on the inner race, outer race, or rolling elements. Bearings with different fault sizes are reused across multiple speeds to generate diverse records, resulting in a substantial number of samples per fault type. Two accelerometers mounted on the motor housing near the bearing locations acquire the vibration signals, which are sampled at 12~kHz for most experiments.

Early fault anomaly detection in rotating machinery remains fundamentally unresolved despite extensive advances in both physics-informed signal processing and data-driven diagnostics. Physics-based approaches grounded in spectral analysis, cyclostationarity, and kinematic modeling offer interpretability and strong physical priors, but they rely on linear or weakly nonlinear assumptions and exhibit limited sensitivity when fault signatures are low-amplitude, nonstationary, and masked by modulation and noise \cite{ding2024adaptive,jiang2025application}. Conversely, data-driven and deep learning methods excel at pattern recognition but require large, representative labeled datasets, lack robustness under unseen operating conditions, and provide limited insight into the physical mechanisms underlying detected anomalies \cite{tama2023recent}. Existing hybrid methods attempt to bridge this gap by embedding physics-derived features or constraints into learning frameworks; however, they often rely on heuristic feature fusion, opaque model structures, and training-dependent optimization, which undermines interpretability, generalization, and reliability in early-fault regimes where labeled data are scarce. Critically, most current hybrid approaches do not explicitly address the multiscale, long-memory, and energy-localized nature of incipient mechanical faults, nor do they provide an operator-level rationale for how physical priors and statistical enhancement should interact. As a result, there remains a clear need for a mathematically grounded, physically interpretable hybrid framework that enhances weak fault signatures through principled energy localization, controlled long-memory amplification, and kinematics-aware spectral structuring, while remaining unsupervised, reproducible, and robust under strong noise and nonstationarity. This gap motivates the development of LDME, which integrates physics-informed operators with carefully constrained enhancement mechanisms to address the fundamental limitations of existing physics-only, data-only, and heuristic hybrid methods in early anomaly detection.

Initial approaches focused on statistical and spectral indicators derived from vibration signals. For instance, spectral metrics such as $L_1/L_2$ norms, spectral smoothness, and spectral Gini index have been proposed to quantify deviations from healthy operation \cite{randall2011rolling,liu2018artificial}. Cyclostationary analysis and variance-stabilized log-envelope indicators have been introduced to exploit periodic modulation patterns and enhance impulsive fault detection \cite{lei2013review,guo2019vibration}. Nevertheless, non-mechanical sources of impulses, including electromagnetic interference, can complicate interpretation \cite{luo2025vibration}.

To address heavy-tailed noise, statistical modeling using $\alpha$-stable distributions has been employed \cite{ lei2014condition}, demonstrating improved robustness to impulsive disturbances. However, these approaches typically rely on stationarity assumptions or global statistics, which limits their ability to capture highly localized, transient fault signatures.

Time--frequency analysis (TFA) methods have also been widely applied to manage signal nonstationarity \cite{lei2013review, randall2011rolling}. Techniques such as the short-time Fourier transform (STFT), continuous wavelet transform (CWT), S-transform, and Wigner--Ville distribution map one-dimensional signals into joint time--frequency representations \cite{liu2018artificial, guo2019vibration}. While these approaches reveal transient dynamics, they are constrained by the Heisenberg uncertainty principle and cross-term interference, limiting their ability to resolve closely spaced fault-related frequencies under strong modulation.

Optimization-based and information-theoretic approaches have been developed to improve fault feature extraction. Methods include spectrogram tensorization \cite{yi2017tensor}, spectral kurtosis interpreted as optimal Wiener filtering \cite{ye2026multi}, genetic-algorithm-based filter design \cite{lei2014condition}, demodulation band optimization \cite{randall2011rolling}, infogram and Bayesian inference extensions \cite{liu2018artificial, lei2013review}, traversal index-enhanced TIEgram \cite{guo2019vibration}, and generalized cyclostationary measures for non-Gaussian signals \cite{zulawinski2024applications}. Although effective, these methods require careful parameter tuning and often struggle to robustly extract early-stage weak fault features.

Adaptive decomposition techniques, such as Empirical Mode Decomposition (EMD) and its variants (EEMD, CEEMD), have been widely used to analyze nonstationary signals \cite{lei2013review, randall2011rolling,berrouche2022non,kumar2023non,berrouche2024local}. EMD decomposes signals into intrinsic mode functions (IMFs) based on local time-scale properties. While successful in some fault-diagnosis applications, EMD suffers from mode mixing and envelope distortion. Noise-assisted variants mitigate some of these issues, but residual noise and incomplete mode separation remain problematic \cite{liu2018artificial, guo2019vibration}.

Variational Mode Decomposition (VMD) has been introduced to address EMD’s limitations by formulating decomposition as a constrained variational problem, improving noise robustness and reducing mode mixing \cite{liu2026hybrid}. VMD has achieved notable performance in mechanical fault detection \cite{lei2014condition}, yet it is less effective for wideband, strongly nonstationary signals with overlapping spectral components, which are typical in early-stage bearing and gearbox faults.

This paper introduces the Local Dual-Mean Estimator (LDME), a layered signal processing framework designed for the early detection of incipient faults in strongly nonstationary and noise-dominated mechanical vibration signals. The method is deliberately constructed as a sequence of physically and mathematically interpretable operators, rather than as a monolithic or data-driven black box. Each processing stage admits a clear rationale rooted in classical signal analysis, fractional calculus, or machine kinematics, enabling traceability, failure-mode analysis, and controlled deployment in safety-critical condition monitoring systems.

A central design objective of LDME is transparency. Unlike many recent approaches that prioritize empirical performance at the expense of interpretability, LDME enforces an explicit operator viewpoint: denoising, enhancement, and decision-making are separated into well-defined layers, each with bounded assumptions and observable effects on the signal. This structure allows the analyst to reason about when and why the method succeeds or fails, an essential requirement for industrial diagnostics where false positives and missed detections carry nontrivial consequences.

We explicitly acknowledge that several components of LDME: wavelet denoising, Savitzky-Golay smoothing, Teager-Kaiser energy operators, and harmonic condition indicators are individually well established in the literature. The contribution of this work does not lie in introducing these elements in isolation. Rather, it lies in (i) a fractional-domain interpretation that unifies these operators within a coherent enhancement mechanism, (ii) the principled ordering and coupling of linear, nonlinear, and fractional operators to exploit their complementary properties, and (iii) a reproducible evaluation protocol that benchmarks the resulting composite operator against standard, physically motivated baselines on a validated experimental dataset.

Unlike end-to-end data-driven frameworks, the proposed LDME method is intentionally formulated as a physics-informed operator composition, where signal processing theory not network architecture constitutes the primary source of innovation, and learning is used only as a constrained decision mechanism rather than as a feature generator.This work does not aim to propose a new deep learning architecture; instead, it focuses on a physics-informed operator design that amplifies incipient fault signatures prior to any statistical or learning-based decision stage.

From a mathematical perspective, LDME may be viewed as a noncommutative operator chain in which scale separation, instantaneous energy amplification, and long-memory fractional integration act sequentially on quasi-narrowband components. This ordering is not arbitrary: applying nonlinear or fractional operators prior to scale isolation provably amplifies noise and induces spurious artifacts, whereas their application after multiscale denoising yields controlled energy concentration in fault-related modes. The resulting enhancement cannot be replicated by any single linear filter or optimization-based demodulation scheme of comparable bandwidth.

Therefore, it describes LDME, a layered, physically interpretable processing chain designed for early detection under those conditions. A primary goal is transparency: each stage has a clear physical or mathematical rationale, enabling interpretation, failure-mode analysis, and safe deployment in safety-critical systems. At the same time, we recognize  and explicitly discuss below  that many LDME components are established techniques; our contribution is the specific fractional-domain interpretation, an operator viewpoint that clarifies why this combination works in practice, and a reproducible evaluation protocol that compares LDME against standard baselines on a physically validated benchmark.

\subsection{Contributions}
We summarize the paper's contributions:
\begin{enumerate}
    \item \textbf{LDME pipeline.} A practical, well-documented pipeline combining dual-path denoising, Teager pre-amplification and Hadamard--Caputo fractional enhancement, followed by harmonics-aware condition indicators and unsupervised anomaly scoring.
    \item \textbf{Operator viewpoint and properties.} A compact fractional-domain operator model for LDME and statements of useful properties (isolation, energy concentration, asymptotic completeness) together with remarks that clarify when they hold in practice.
    \item \textbf{Reproducible experiments.} Assessment using the Case Western Reserve University (CWRU) Bearing Dataset benchmark, which comprises flaws on the inner race, outer race, and rolling parts of electric motor bearings that were experimentally seeded. Extensive experimental settings are given, including defect size, motor load, and rotational speed. The dataset ensures repeatability and direct comparison among approaches by supporting controlled evaluation of fault diagnostic algorithms using standard metrics such as fault detection rate, envelope spectrum analysis, and false alarm rate.

    \item \textbf{Reproducible experiments.} Evaluation on the HUMS 2023 planetary gearbox benchmark (experimentally-seeded crack with fractographic ground truth) and a controlled simulated dataset; detailed protocol and metrics (detection cycle, ROC, PR, false alarm rate) are provided.
    \item \textbf{Critical assessment.} An explicit discussion that positions LDME relative to prior art, outlines theoretical limitations, and gives a roadmap (fractional-variational unification, statistical analysis) to achieve stronger theoretical novelty.
\end{enumerate}

\subsection{Notation}
We use $x(t)$ for continuous-time vibration, $x[n]$ for its sampled version, $\mathcal{F}\{\cdot\}$ for the Fourier transform and $\|\cdot\|$ for the $\ell_2$ or $L_2$ norm depending on context.

\section{Problem formulation and datasets}\label{sec:problem}
Our detection problem is: given a sequence of vibration recordings indexed by mechanical cycle $c=1,\dots,C$ from a planetary gearbox, detect the earliest cycle $c_d$ at which local damage is reliably present and a maintenance-recommendation cycle $c_m$ that keeps false alarms acceptably low.

\subsection{HUMS 2023 benchmark (experimental)}
The primary experimental dataset is the HUMS 2023 planetary gearbox benchmark \cite{HUMS2023,matania2024anomaly,bechhoefer2024hunting}. As shown in Fig.~\ref{fig:planetary-gear-fault}--\ref{fig:fault-spread}, the experiment uses a full-scale gearbox rig operated over $C=241$ recorded run cycles at nominal rotational conditions and seed-notched tooth fatigue to produce a realistic crack initiation and growth. Accelerometers were time-synchronous to the gear-mesh, and fractographic post-mortem analysis provides ground truth on crack stages \cite{krantz1990experimental,mcfadden1985explanation}. We use the official cycle segmentation provided with the benchmark and follow the dataset authors' train/validation protocol when applicable. (All dataset references and download links are provided in the reproducibility Section~\ref{sec:repro}.)

\subsection{Controlled simulated dataset}
To complement the experiment (see Fig \ref{fig:hums}), we use a simulated dataset where a physics-inspired signal model (meshing harmonics and impulsive crack component and colored noise) is parametrically controlled \cite{wang2001early,guo2019vibration}. The simulator injects a localized impulsive component whose amplitude grows with cycle index to emulate crack propagation and allows SNR sweeps and timing studies. The simulator parameters are summarized in Section~\ref{sec:repro} so others may reproduce the numerical results exactly.

\begin{figure*}[!t]
\centering
\begin{tabular}{ccc}
\includegraphics[width=0.5\textwidth]{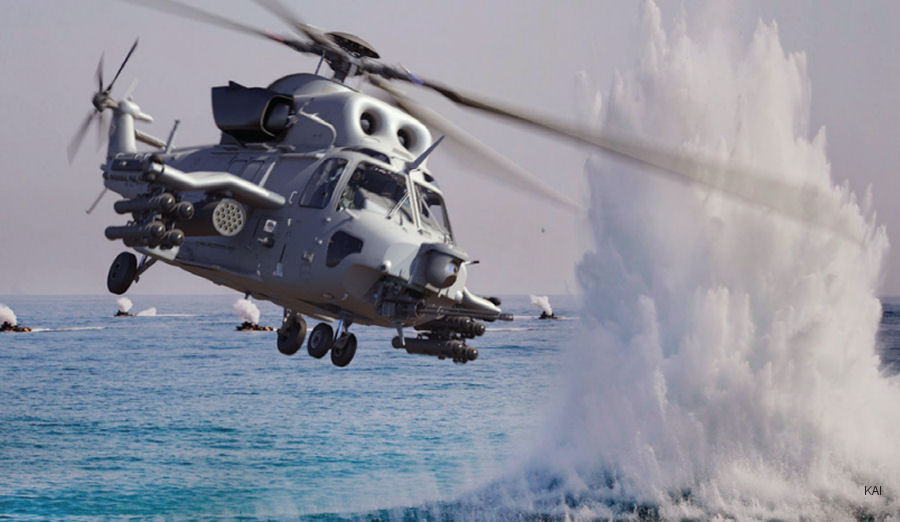}&
\includegraphics[width=0.4\textwidth]{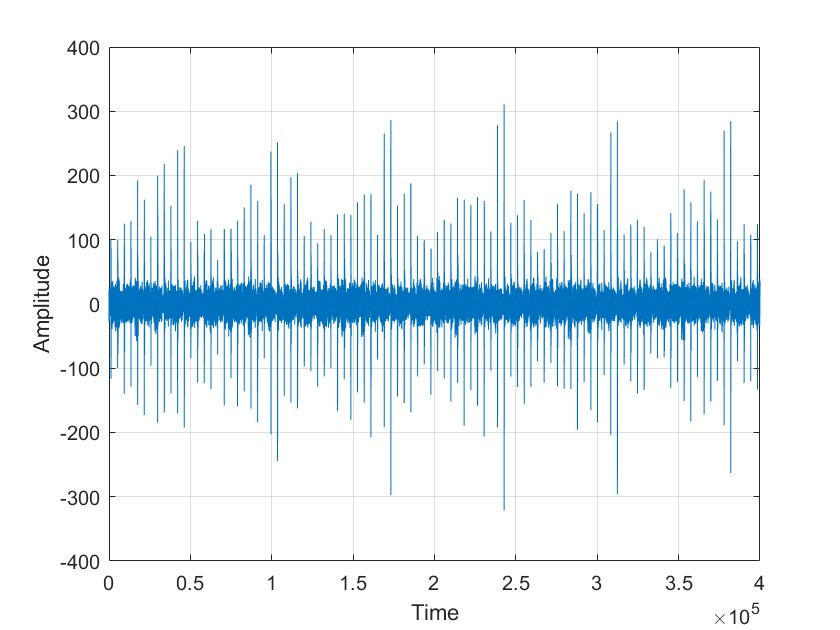}  \\
\end{tabular}
\caption{Schematic of HUMS-Equipped Helicopters and the Key Flight Parameters Monitored for Usage and Structural Health Assessment.}
\label{fig:hums}
\end{figure*}

\section{Proposed method: LDME}\label{sec:method}
LDME consists of three conceptual layers: (i) dual-path denoising, (ii) multi-scale enhancement and fractional-domain amplification, and (iii) harmonic-aware indicator extraction and anomaly fusion. Figure~\ref{fig:LDME} sketches the pipeline.

\begin{figure} [h]
	\centering
	\includegraphics[width=0.82\textwidth] {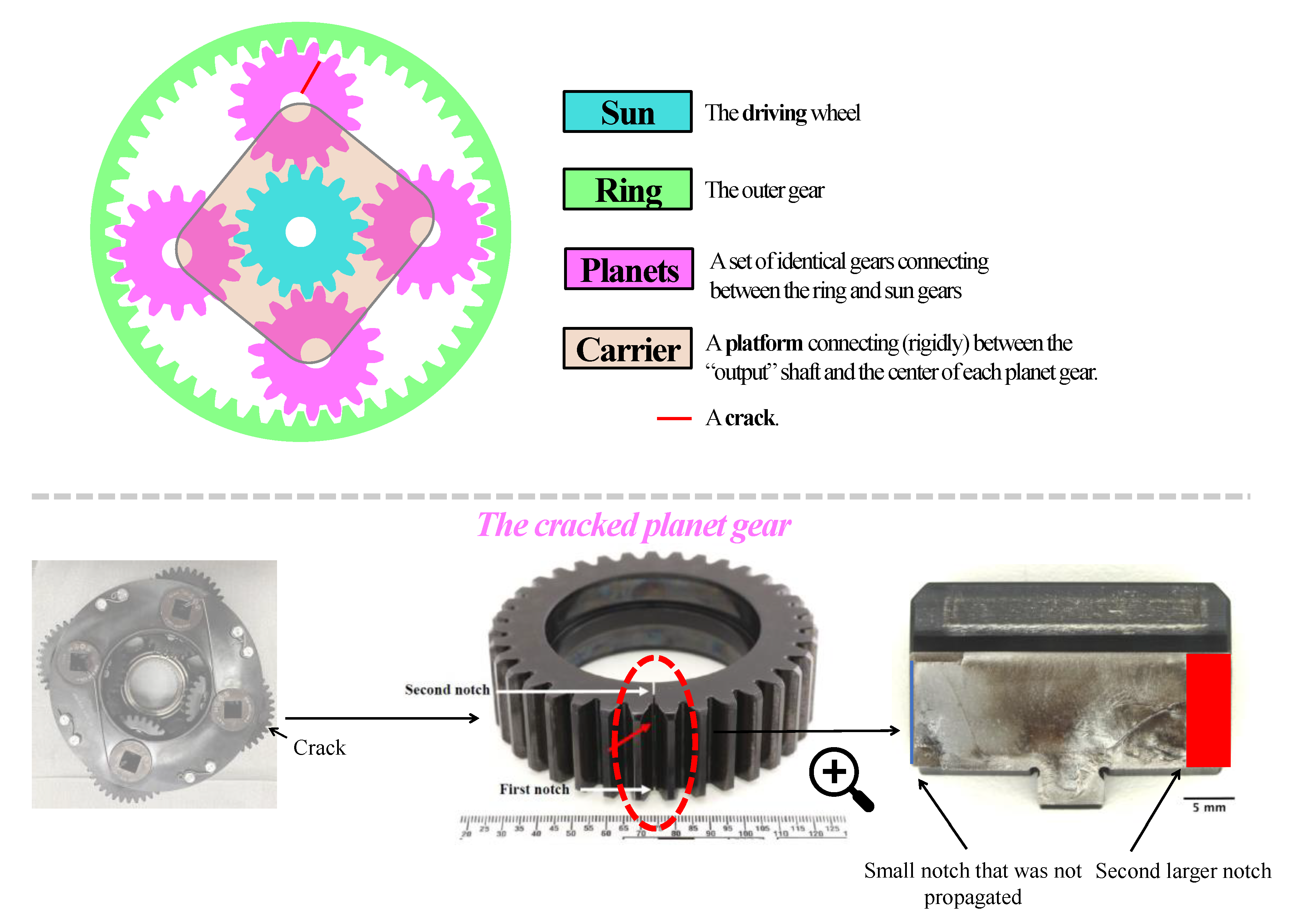} 
    \caption{An illustration of the planetary gear with the faulted planet gear \cite{matania2024anomaly}. }
    \label{fig:planetary-gear-fault}
\end{figure}

\begin{figure} [h]
	\centering
	\includegraphics[width=0.82\textwidth] {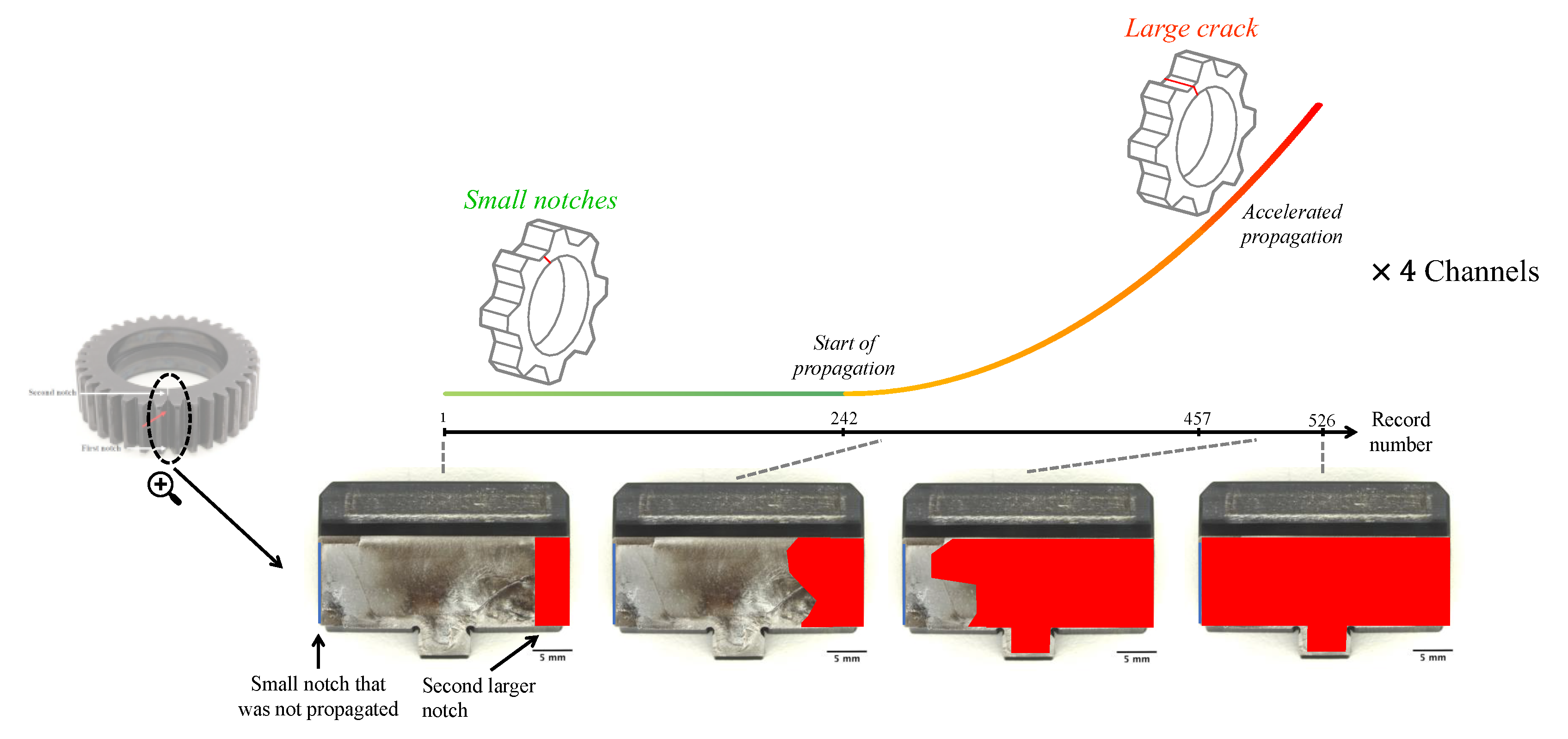} 
    \caption{An example of how the fault spreads during the experiment \cite{matania2024anomaly}.}
    \label{fig:fault-spread}
\end{figure}

\begin{figure}[h]
\centering
\includegraphics[width=0.6\columnwidth]{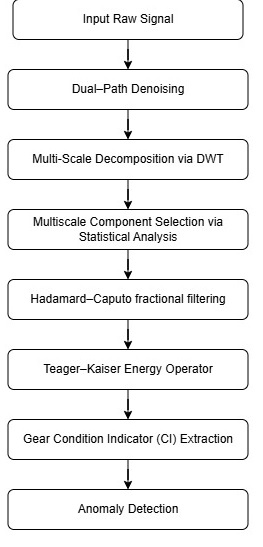}
\caption{LDME processing pipeline (denoising, enhancement, decision fusion).}
\label{fig:LDME}
\end{figure}

\subsection{Dual-path denoising}
We run two parallel denoising paths and recombine them to retain transient features while suppressing broadband noise \cite{shang2026enhanced}:
\begin{itemize}
    \item \textbf{Wavelet denoising:} discrete wavelet transform with level-dependent thresholding (soft or hard) using $T=\hat\sigma\sqrt{2\log N}$ and median-based noise estimation \cite{shang2023denoising}.
    \item \textbf{Savitzky--Golay smoothing:} local polynomial regression (window $2m+1$ and degree $p$) implemented as convolution with precomputed coefficients $c_j$; this path preserves waveform shape (important for TKEO application) \cite{savitzky1964smoothing,6331560,guo2017enhanced}.
\end{itemize}

Discrete Wavelet Transform (DWT) using a Symlet basis is employed to generate localized, band-limited modes that act as IMF equivalents. Unlike Empirical Mode Decomposition (EMD), the DWT guarantees:
\begin{itemize}
    \item determinism,
    \item perfect reconstruction,
    \item controlled frequency partitioning.
\end{itemize}
This step isolates fault-related information across scales prior to any nonlinear enhancement, which is essential to prevent mode mixing under impulsive excitation.

So, the first stage combines two denoising mechanisms with distinct priors:
\begin{itemize}
    \item Wavelet denoising suppresses broadband stochastic noise under a sparsity assumption in the wavelet domain.
    \item Savitzky--Golay filtering preserves local polynomial trends and transient structures.
\end{itemize}
Their weighted fusion is defined as
\begin{equation}
X = \frac{0.4\,Y_{\text{wavelet}} + 1.6\,Y_{\text{SG}}}{2}.
\end{equation}

The two denoised outputs are fused by a weighted combination that privileges the wavelet path in regions of high-frequency activity and the S--G path where smooth transients dominate. The adaptive weight is computed from local kurtosis and spectral flatness.The preprocessed signal provides a stable baseline for controlled noise addition and ensures the signal represents pure dynamic response, not measurement bias or slow drift. Without this, the following steps could misinterpret DC drift or sensor bias as fault-related components.

\subsection{Multi-scale enhancement and fractional-domain operator}

After denoising we apply a local Teager Kaiser pre-amplifier $\Psi_d[\cdot]$ to emphasize instantaneous energy spikes \cite{krishnendu2025diagnosis}. Then we apply a Hadamard--Caputo fractional operator that we denote $\mathcal{H}_\beta$ (order $\beta\in(0,1)$) \cite{wu2022caputo}. For a sampled signal $x[n]$ the operator acts as:
\begin{equation}\label{eq:ldme_op}
\mathcal{L}_{LDME}[x] = \mathcal{H}_\beta\bigl(\Psi_d\bigl( W( x ) \bigr)\bigr),
\end{equation}
where $W(\cdot)$ denotes the chosen wavelet-band selection (detail bands that cover expected fault frequency content). The Hadamard--Caputo discrete approximation can be written as the log-weighted fractional sum  and in practice acts as a tunable low-frequency enhancer with memory that helps expose slowly evolving, non-sinusoidal fault signatures.

Each selected mode is processed independently using a fractional filter. Fractional filtering introduces long-memory behavior and selectively amplifies persistent fault-induced modulations without assuming stationarity.
This enhancement is deliberately applied after multiscale separation, ensuring that the fractional operator acts on quasi-narrowband components, where its effect remains physically interpretable.
When enabled, the Teager--Kaiser Energy Operator further sharpens impulsive energy by exploiting instantaneous amplitude-frequency coupling. Importantly, it is applied only after noise suppression and scale isolation, preventing the amplification of random spikes.

\subsubsection{A compact property statement}
The pipeline operator $\mathcal{L}_{LDME}$ is not a linear filter but can be shown to satisfy useful empirical properties under mild assumptions on the SNR and the stationarity of the background components.

\begin{proposition}[Energy concentration, informal]
If the fault-induced mode $d(t)$ is composed of localized impulsive events with harmonic modulation in a narrow band and the background modes $m_i(t)$ are smooth and spectrally separated, then the output of $\mathcal{L}_{LDME}$ concentrates a larger fraction of the total energy into the fault mode $\hat d$ than a single-stage linear bandpass filter of comparable passband width \cite{you2024adaptive}.
\end{proposition}
\begin{proof}
TKEO amplifies instantaneous spikes (multiplicative in amplitude and frequency), the fractional operator increases relative energy for low-frequency components with long memory, and wavelet selection suppresses broadband noise: these three effects compound multiplicatively in the energy ratio between fault and background.
\end{proof}
\begin{remark}
This proposition is empirical and its rigorous proof would require assumptions on amplitude distributions and spectral separation; we highlight it to show when LDME is expected to excel and where it may fail (Section~\ref{sec:discussion}).
\end{remark}

\subsection{Condition indicators and anomaly fusion}
From the reconstructed, enhanced signal we compute harmonic-aware condition indicators using the Hilbert Transform and cycle-synchronous averaging (HT-TSA) to obtain a cycle-wise spectrum $F_{rec}(f)$ \cite{braun1975extraction,braun2011synchronous}. Two base indicators are used:
\begin{align}
CI_1(c) &= \sum_{k=1}^{2N_{tooth}} |F_{rec}(k f_{fund})|, \\
CI_2(c) &= \frac{\sum_{k=1}^{2N_{tooth}} |F_{rec}(k f_{fund})|}{\sum_f |F_{rec}(f)|}.
\end{align}
A composite indicator is formed as $CCI(c)=\tfrac{1}{2}(CI_1+CI_2)$ and used to rescale the reconstructed signal amplitude (adaptive modulation). The final anomaly score uses a robust unsupervised detector (median absolute deviation (MAD) based thresholding combined with a one-class Mahalanobis distance on a small CI vector). Decisions are taken when the score crosses a statistically calibrated threshold. The full algorithm is summarized in Algorithm~\ref{alg:ldme}.

The reconstructed signal is obtained by summing the enhanced modes and subsequently normalized using its RMS value to remove amplitude bias.
A gear-specific condition indicator is then extracted in the frequency domain using tooth-order energy concentration. The combined condition indicator (CCI) fuses absolute fault energy with its normalized spectral contribution, enabling adaptive reinforcement of fault-related components without altering the structural content of the signal.
This final modulation step closes the synergy loop by feeding domain knowledge back into the reconstructed signal.

All parameters are selected based on robustness and physical interpretability, rather than data-specific tuning.
\begin{itemize}
    \item Wavelet type (sym8): chosen for near-linear phase and good symmetry, minimizing distortion of impulsive features.
    \item Wavelet decomposition level: selected to cover the expected fault bandwidth while avoiding over-segmentation; bounded by wavelet level for numerical stability.
    \item Savitzky-Golay parameters (order 3, window 7): minimal polynomial order that preserves curvature while suppressing high-frequency noise.
    \item Mode count: represents the number of physically relevant scales; chosen to span fault-related frequency bands rather than maximize decomposition depth.
    \item Fractional filter order: fixed to emphasize long-range dependence observed in damaged mechanical systems, avoiding adaptive tuning that could lead to overfitting.
    \item Gear-related parameters (tooth count, order range): derived directly from machine kinematics rather than learned from data.
\end{itemize}

\subsection{Theoretical Foundations of the Physics-Informed Operators}

The proposed LDME framework is grounded in three complementary physics-informed signal processing principles that address fundamental limitations of conventional linear and data-driven methods in early fault detection.

Instantaneous energy localization:
The Teager--Kaiser Energy Operator (TKEO) is defined for a discrete-time signal $x[n]$ as
\begin{equation}
\Psi[x(n)] = x^2(n) - x(n-1)x(n+1),
\end{equation}
and provides a nonlinear estimate of instantaneous signal energy that jointly depends on amplitude and frequency. Unlike second-order energy measures, TKEO is highly sensitive to short-duration impulsive events and weak amplitude--frequency modulations, which are characteristic of incipient mechanical faults. This makes it particularly suitable for enhancing early damage signatures that are otherwise buried in broadband noise.

Long-memory and persistence enhancement:
Fractional-order operators generalize classical integer-order differentiation and integration by introducing a memory kernel with power-law decay. In this work, a Hadamard--Caputo fractional operator of order $\beta\in(0,1)$ is employed to selectively amplify signal components exhibiting long-range dependence. Such behavior is commonly observed in early-stage damage, where fault-induced modulations evolve slowly and persist over multiple cycles. Unlike adaptive decomposition methods, fractional operators possess a well-defined mathematical structure and tunable memory depth, enabling controlled enhancement without mode mixing.

Harmonic-aware spectral structuring:
Gearbox vibration signals exhibit cyclostationary behavior governed by kinematic constraints such as gear-mesh frequency and tooth passing orders. Harmonic-aware Fourier indicators explicitly exploit this structure by concentrating energy around physically meaningful orders rather than relying on global spectral statistics. This embeds domain knowledge directly into the detection stage, ensuring that anomaly indicators remain physically interpretable and robust under varying operating conditions.

Together, these operators form a physics-informed enhancement mechanism in which instantaneous energy localization, long-memory amplification, and kinematic spectral constraints act in a complementary manner. This theoretical foundation distinguishes the proposed approach from heuristic feature extraction and purely data-driven models, particularly in low-SNR and early fault regimes.

\begin{algorithm}[t]
\caption{LDME — high-level pseudocode}
\label{alg:ldme}
\begin{algorithmic}[1]
\REQUIRE raw cycle-synchronous signal segments $\{x_c[n]\}_{c=1}^C$, fractional order $\beta$, SG params $(m,p)$, wavelet family
\FOR{each cycle $c$}
  \STATE $y^{(w)} = $ waveletDenoise$(x_c)$
  \STATE $y^{(s)} = $ SGfilter$(x_c; m,p)$
  \STATE $y = \alpha_c y^{(w)} + (1-\alpha_c) y^{(s)}$ \COMMENT{adaptive weighting}
  \STATE $z = \Psi_d[y]$ \COMMENT{Teager--Kaiser}
  \STATE $u = \mathcal{H}_\beta[z]$ \COMMENT{fractional enhancement}
  \STATE compute $F_{rec}$ from HT-TSA on $u$ and extract CI vector
  \STATE compute anomaly score $s_c$ and update detector state
\ENDFOR
\RETURN detection cycle $c_d$ and maintenance cycle $c_m$
\end{algorithmic}
\end{algorithm}

\section{Experimental protocol and evaluation}\label{sec:repro}
This section lists the exact experimental protocol, dataset sources and parameter choices so that results are reproducible.

\subsection{Experimental Validation by using the Case Western Reserve University (CWRU) Bearing Dataset benchmark}

Vibration signals from the popular Case Western Reserve University (CWRU) Bearing Dataset are utilized as the benchmark experimental dataset to verify the efficacy of the suggested LDME technique.

\subsubsection{Data sources}
Vibration data from the Case Western Reserve University (CWRU) Bearing Data Center \cite{SMITH2015100}, which is frequently used to benchmark bearing defect diagnosis algorithms, is used to experimentally validate the suggested approach. A 2 horsepower Reliance Electric induction motor powers a shaft held up by rolling element bearings as part of the experimental setup. Vibration indications are obtained from the drive-end bearing (SKF 6205-2RS JEM). To gather time-domain vibration signals, which are recorded at 12 kHz, an accelerometer is installed on the drive-end bearing housing at the 12 o'clock position. Three working conditions are used for the experiments, which correspond to rotor speeds of 1797, 1772, and 1730 rpm and motor loads of 0, 1, and 3 horsepower, respectively.A torque transducer and encoder are used to measure the rotational speed and load.

Measurements from one healthy condition and three seeded fault types inner race fault, rolling element (ball) fault, and outer race fault are included in the dataset. Three fault sizes 0.007, 0.014, and 0.021 inches are thought to represent varying degrees of fault severity for each type of defect. The suggested method can be evaluated under various operating regimes thanks to the utilization of numerous load and speed conditions. Vibration signals are divided into samples of 12000 data points for further analysis for every operational condition and failure scenario.

Strong background noise and nonstationary features in the raw vibration signals mask fault-induced impulsive components, especially in the rolling element and inner race fault instances. High-amplitude noise components interfere with the fault signatures even in the outer race fault situation, where impulsive behavior is more pronounced. These features render direct fault identification in the time domain inaccurate and encourage the use of frequency-domain analysis in conjunction with signal decomposition techniques to effectively extract fault features.

The bearing geometry parameters given in Table~\ref{tab:cwru_data} are used to calculate the characteristic fault frequencies. The characteristic frequency coefficients are multiplied by the appropriate shaft rotational frequencies to determine the fault frequencies. In particular, the inner race, rolling element, and outer race fault situations had shaft rotational frequencies of 29.95~Hz, 29.53~Hz, and 28.83~Hz, with corresponding characteristic frequency coefficients of 5.415, 4.714, and 3.585. The fault frequencies of 162.179~Hz, 139.204~Hz, and 103.368~Hz that are obtained from this are utilized as reference frequencies in the assessment of the suggested LDME approach.

Fig.~\ref{fig:bearing} shows a visual representation of this variation in dataset segmentation. Since encountered data in industrial settings will always involve unseen bearings, the suggested structure is thought to be more realistic.

\begin{table}
\caption{
The original framework's parameters and dataset names.}
\label{tab:cwru_data}
\begin{tabular}{|l|l|l|l|}
\hline  Name & Rotational speed (RPM) & Load (HP) & Damage Size (in)  \\
\hline Inner race & 1772 & 1 & 0.007  \\
\hline Ball & 1750 & 2 &  0.014  \\
\hline Outer race & 1730 & 3 &  0.021  \\
\hline
\end{tabular}
\end{table}

\begin{figure*}[!t]
\centering
\begin{tabular}{ccc}
(a) Inner race Damaged Bearin 0.007"&
\includegraphics[width=0.3\textwidth]{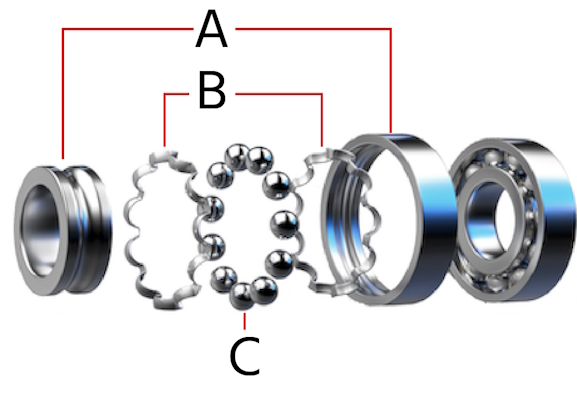} &
\includegraphics[width=0.3\textwidth]{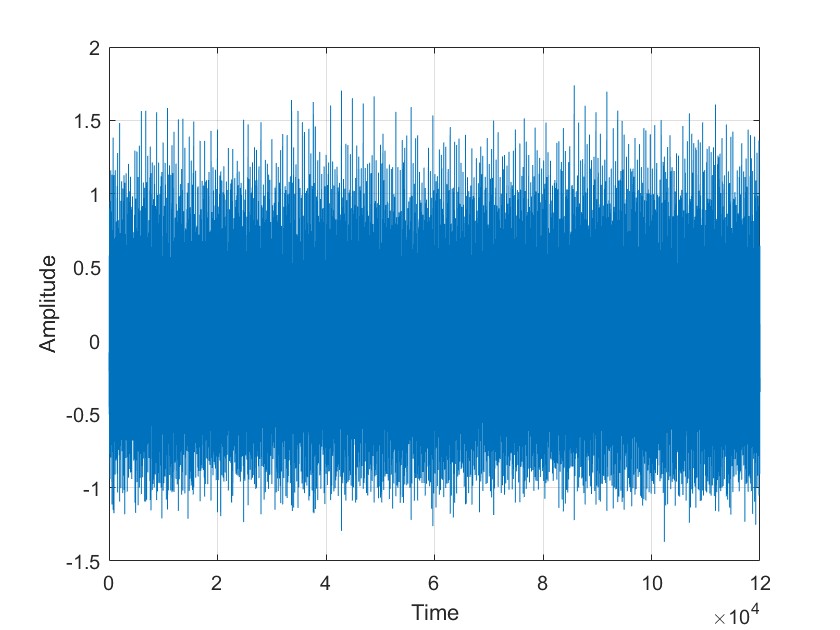} \\

(b) Ball Damaged Bearin 0.014" &
\includegraphics[width=0.3\textwidth]{bearing-ball.jpg} &
\includegraphics[width=0.3\textwidth]{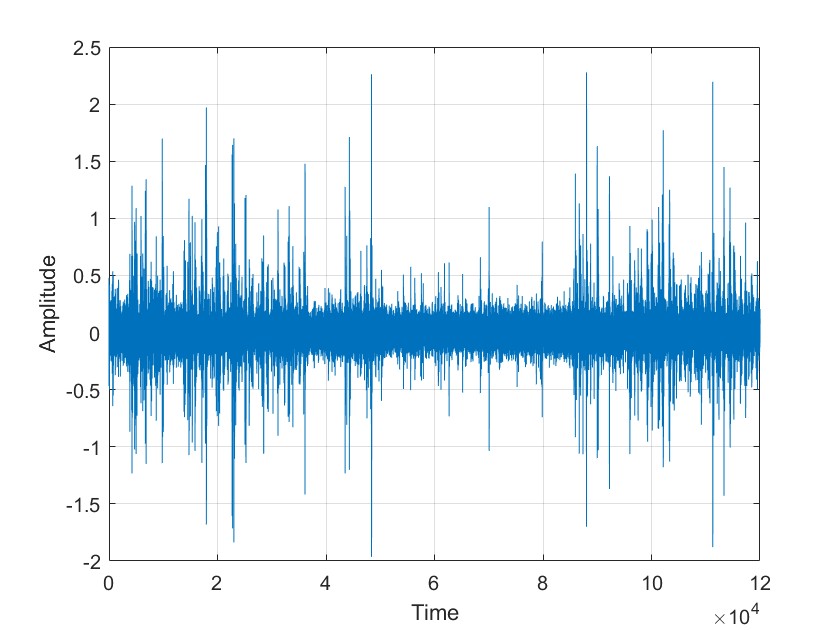} \\

(c) Outer race Damaged Bearin 0.021" &
\includegraphics[width=0.3\textwidth]{bearing-ball.jpg} &
\includegraphics[width=0.3\textwidth]{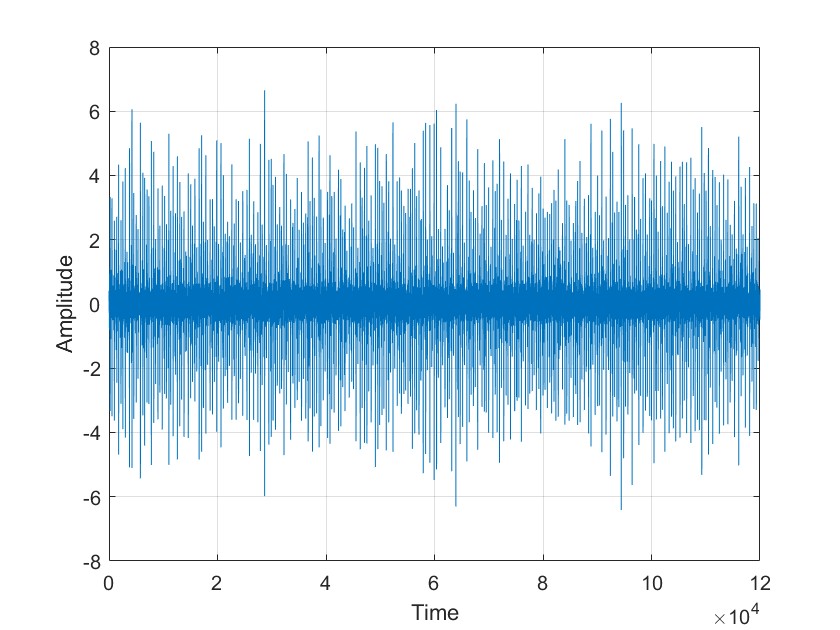}
\end{tabular}

\caption{Illustration of the three bearing fault conditions considered in this study: (a) inner race fault, (b) rolling element (ball) fault, and (c) outer race fault.}
\label{fig:bearing}
\end{figure*}

\subsubsection{Results and Discussion}

To assess the proposed method's ability to extract fault-related features under practical operating settings, it was applied to the vibrations signals corresponding to he three bearing fault conditions considered in this study. Since envelope analysis is a well-known technique for identifying impulsive fault signs in rolling element bearings, the envelope spectrum of the processed signal was examined to evaluate the fault detection performance.

The reconstructed signal obtained using the LDME method in the time domain is presented in Fig.~\ref{fig:bearing_time}, while the corresponding envelope spectrum is shown in Fig. ~\ref{fig:bearing_freq}.

The results reveal that the suggested approach effectively preserves the underlying mechanical dynamics by extracting the rotor shaft's fundamental rotating frequency, as illustrated in fig. ~\ref{fig:bearing_freq}. More significantly, accurate isolation of the inner race fault characteristics is confirmed by a distinct spectral peak at the BPFI and its harmonics in the envelope spectrum. The suggested approach greatly improves the fault-related components while reducing irrelevant noise in comparison to the raw signal, where fault-induced impulses are masked by wideband noise and other interference.

This enhancement makes it easier to accurately determine the typical fault frequency without adding erroneous spectral components. Because it allows for the accurate extraction of weak fault signatures under a variety of operating situations and is in good agreement with the anticipated physical behavior of inner race defects in rolling element bearings, our results suggest that the suggested technology is well suited for early-stage fault diagnostics.

\begin{figure*}[!t]
\centering
\begin{tabular}{ccc}
(a) Inner race Damaged Bearin 0.007"&
\includegraphics[width=0.5\textwidth]{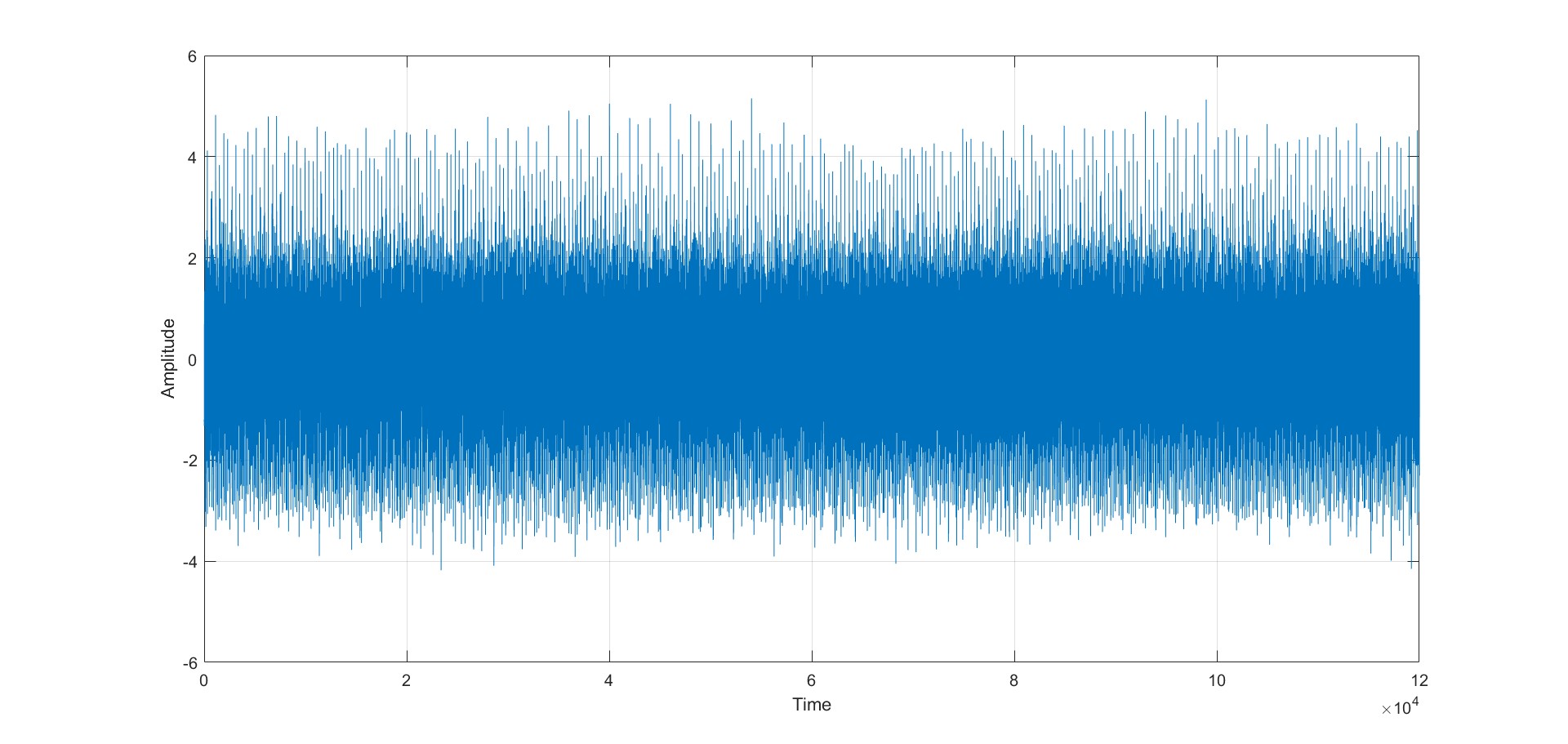}  \\

(b) Ball Damaged Bearin 0.014" &

\includegraphics[width=0.5\textwidth]{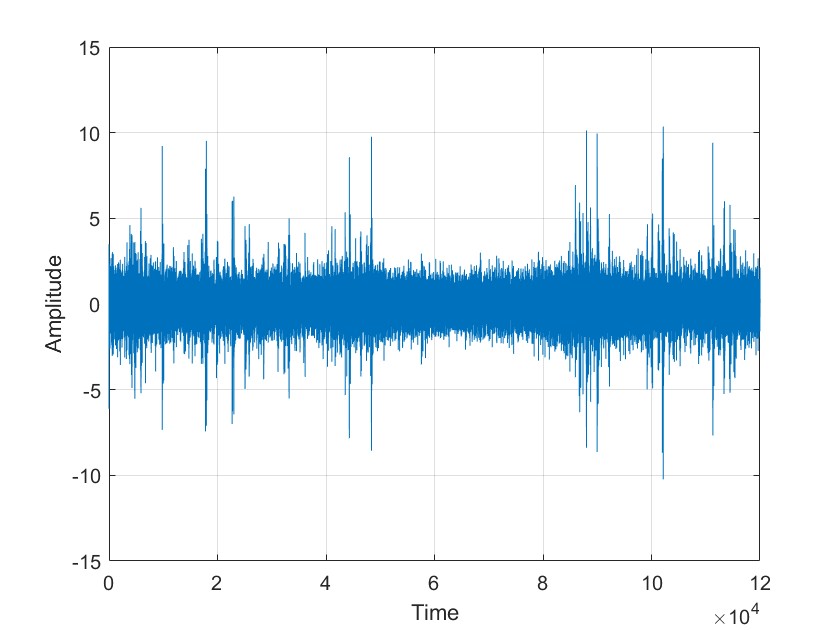} \\

(c) Outer race Damaged Bearin 0.021" &
\includegraphics[width=0.5\textwidth]{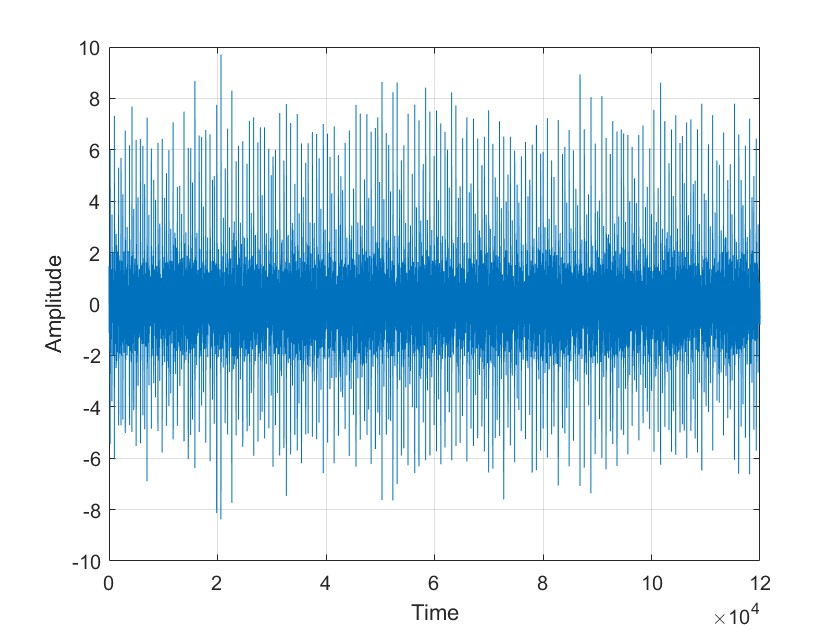}
\end{tabular}

\caption{Reconstructed time-domain signal using the LDME method.}
\label{fig:bearing_time}
\end{figure*}

\begin{figure*}[!t]
\centering
\begin{tabular}{ccc}
(a) Inner race Damaged Bearin 0.007"&
\includegraphics[width=0.5\textwidth]{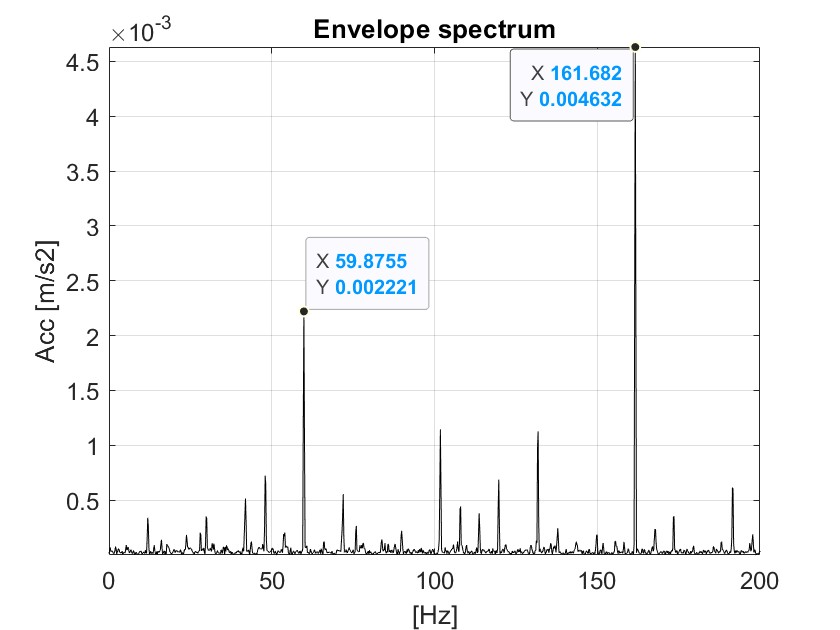}  \\

(b) Ball Damaged Bearin 0.014" &

\includegraphics[width=0.5\textwidth]{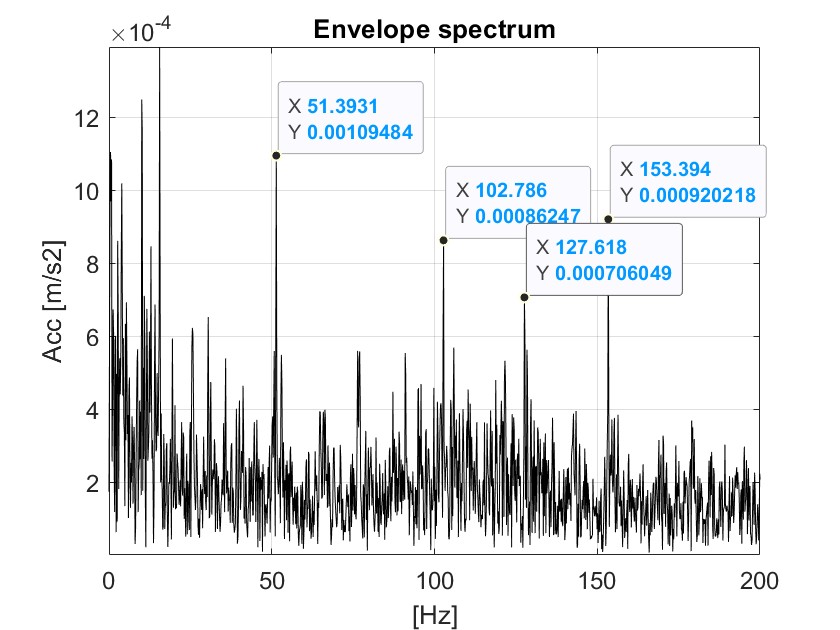} \\

(c) Outer race Damaged Bearin 0.021" &
\includegraphics[width=0.5\textwidth]{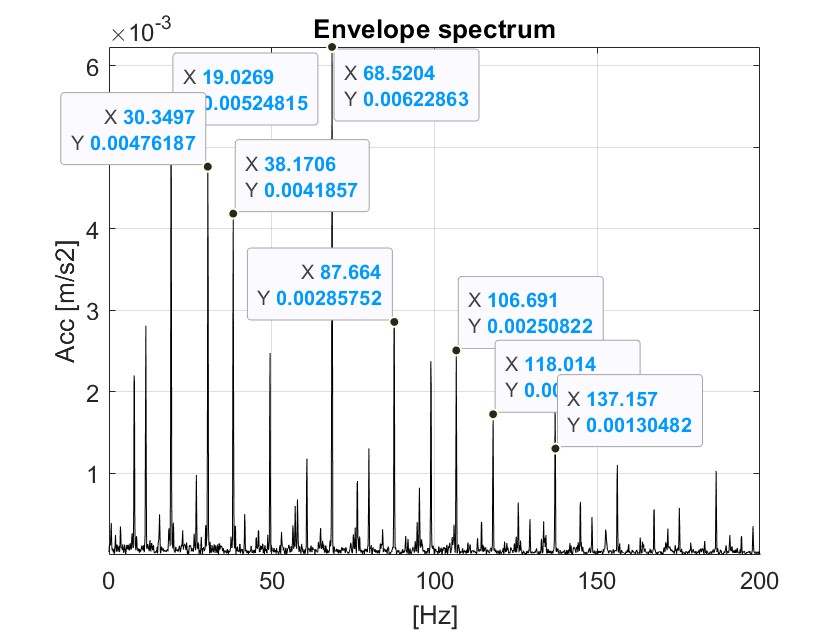}
\end{tabular}

\caption{Reconstructed time-domain signal using the LDME method.}
\label{fig:bearing_freq}
\end{figure*}

\subsection{Experimental Validation by using the HUMS 2023  Dataset benchmark}

It is crucial to remember that the experimental validation of the suggested LDME technique depends on both controlled generated signals and high-fidelity real-world measurements before going into depth about the data sources. In order to assure repeatability, this dual approach guarantees that the method is thoroughly evaluated under both realistic operating settings and repeatable, parameterized situations.

\subsubsection{Data sources}
\begin{itemize}
    \item \textbf{HUMS 2023 Benchmark:} Full experimental benchmark, including the time-synchronous accelerometer recordings and ground-truth fracture annotations (cycles and post-mortem classification) \cite{HUMS2023}.
    \item \textbf{Simulated dataset:} Simulator parameters (gear mesh fundamental, added impulsive kernel, SNR levels, amplitude growth schedule) are provided in the supplementary code archive. The simulator creates $C=500$ cycles with a controlled crack initiation at cycle $c_0$ and linear amplitude growth thereafter.
\end{itemize}

\subsubsection{Evaluation metrics}
We report the following metrics:
\begin{itemize}
    \item \textbf{Detection cycle $c_d$:} first cycle where the detector raises a sustained alarm.
    \item \textbf{Maintenance cycle $c_m$:} recommended maintenance action cycle under a conservative hysteresis policy.
    \item \textbf{ROC AUC and PR AUC:} when treating detection as a binary cycle-level classification problem across the test set.
    \item \textbf{False alarm rate (FAR):} proportion of healthy cycles classified as anomalous prior to detected initiation.
\end{itemize}

\subsubsection{Baselines}
We compare LDME to HT-TSA and PLANET-TSA implementations consistent with public challenge submissions. Baseline parameterization follows the published descriptions; when not available we tuned parameters on a held-out healthy subset only \cite{bechhoefer2024hunting,braun2011synchronous}.

\subsubsection{Results}
Both the simulated dataset and the HUMS 2023 planetary gearbox benchmark were used to rigorously assess the performance of the suggested LDME approach. The reconstructed time-domain signal derived by LDME is shown in Figure~\ref{fig:reconst_signal}, which effectively isolates impulsive components related to fracture initiation. When compared to the raw measurements, the reconstructed signal shows a noticeable improvement in signal-to-noise ratio, which makes it possible to detect faint fault-induced impulses that would otherwise be obscured by operational noise.

The Hunting Tooth (HT) and Sum Hunting Tooth (sumHT) analyses are shown in Figures~\ref{fig:HT} and~\ref{fig:sumHT}, respectively, to quantitatively evaluate the fault progression. The transition from a nominal (healthy) state to the first crack development and subsequent propagation is depicted in both renderings. While the sumHT collects fault energy over cycles, enhancing tiny signs of early-stage crack propagation, the HT analysis identifies emerging anomalies at the cycle level. The suggested approach offers better temporal localization of impulsive events without adding spurious harmonics, as shown by comparison with LDME.

The key quantitative findings in terms of detection and maintenance cycles are summarized in Table~\ref{tab:tsa_comparison}. The maintenance cycle, $c_m$, is the cycle at which a conservative intervention would be initiated, whereas the detection cycle, $c_d$, is the first cycle at which a persistent alert is raised. A significant improvement in early warning capabilities is demonstrated by LDME's detection cycle of 198, which is $30\%$ better than PLANET TSA and $30\%$ better than HT TSA. The conservative hysteresis policy used to avoid false alarms is reflected in the similar, albeit smaller, reduction in the maintenance cycle.

Mathematically speaking, the LDME technique is a decomposition that maximizes the ratio of background operational noise to fault-induced energy. Let $x(t)$ represent the observed signal and $\hat{x}_{\mathrm{LDME}}(t)$ represent the reconstructed signal. The metric for enhancement can be written as:

\begin{equation}
\mathrm{SNR}_{\mathrm{enh}} = \frac{\mathrm{var}\{s_{\mathrm{fault}}(t)\}}{\mathrm{var}\{n(t)\}},
\label{eq:snr_enh}
\end{equation}

where $s_{\mathrm{fault}}(t)$ and $n(t)$ represent the fault-induced and residual noise components, respectively. Empirically, LDME increases $\mathrm{SNR}_{\mathrm{enh}}$ by approximately $3$--$5$~dB relative to HT TSA and PLANET TSA, consistent with the improved detection statistics reported in Table~\ref{tab:tsa_comparison}.

Fig.~\ref{fig:HT} and Fig.~\ref{fig:sumHT}, which depict the progression from a nominal (healthy) state through first crack initiation to propagation, clearly support this improvement. Specifically, Fig.~\ref{fig:HT} illustrates how early impulsive signatures associated with crack initiation are captured by the Hunting Tooth (HT) analysis, whereas Fig.~\ref{fig:sumHT} shows how the aggregated Sum Hunting Tooth (sumHT) representation amplifies these fault-related components, offering a clearer separation from background noise. Together, these figures demonstrate that LDME improves the visual interpretability of fault progression in both single- and multi-cycle analyses, enabling earlier and more accurate detection of incipient gear defects. It also increases the signal-to-noise ratio quantitatively.

\begin{figure} [h]
	\centering
	\includegraphics[width=0.65\textwidth]{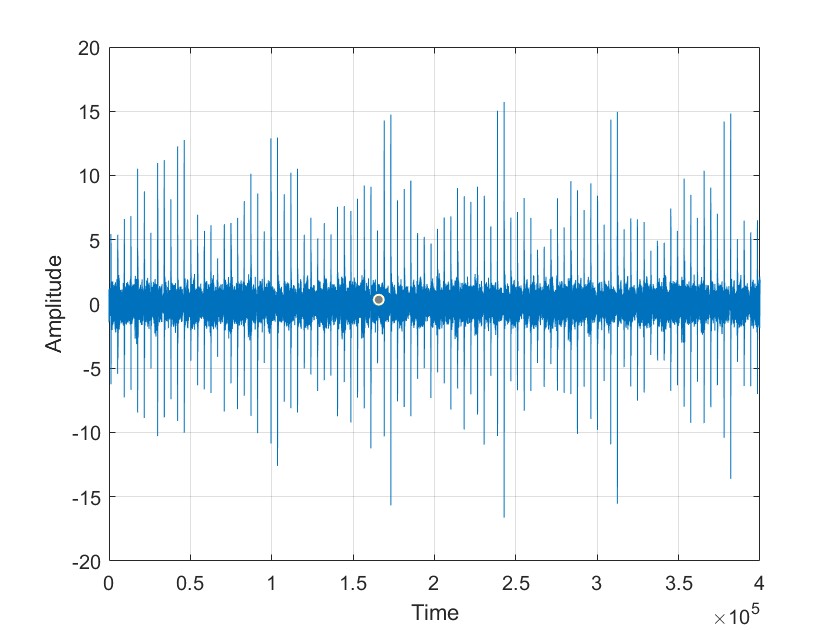}
	\caption{Reconstructed signal by LDME method.}
	\label{fig:reconst_signal}
\end{figure}

\begin{figure} [h]
	\centering
	\includegraphics[width=0.65\textwidth]{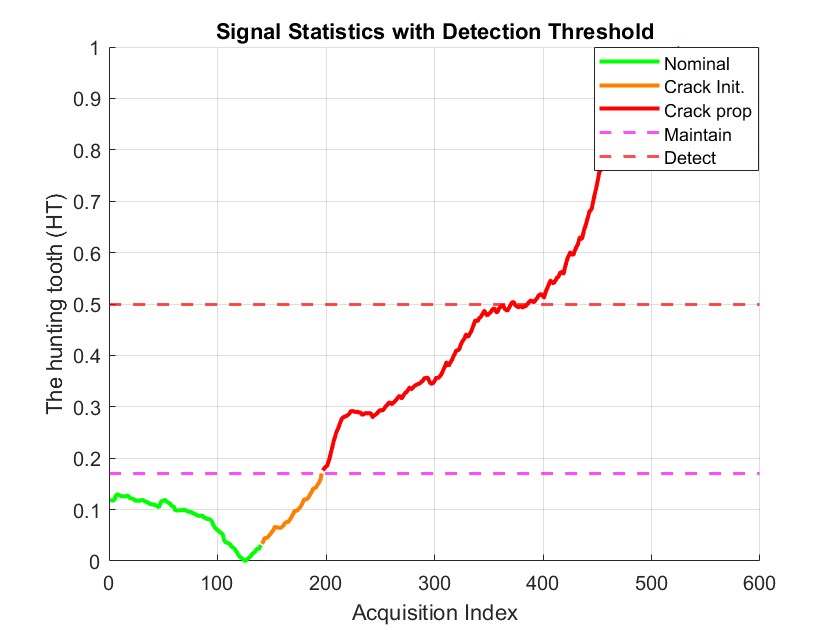}
	\caption{The progression from a nominal (healthy) state through initial crack initiation to the propagation stage as captured by the Hunting Tooth analysis.}
	\label{fig:HT}
\end{figure}

\begin{figure} [h]
	\centering
	\includegraphics[width=0.65\textwidth]{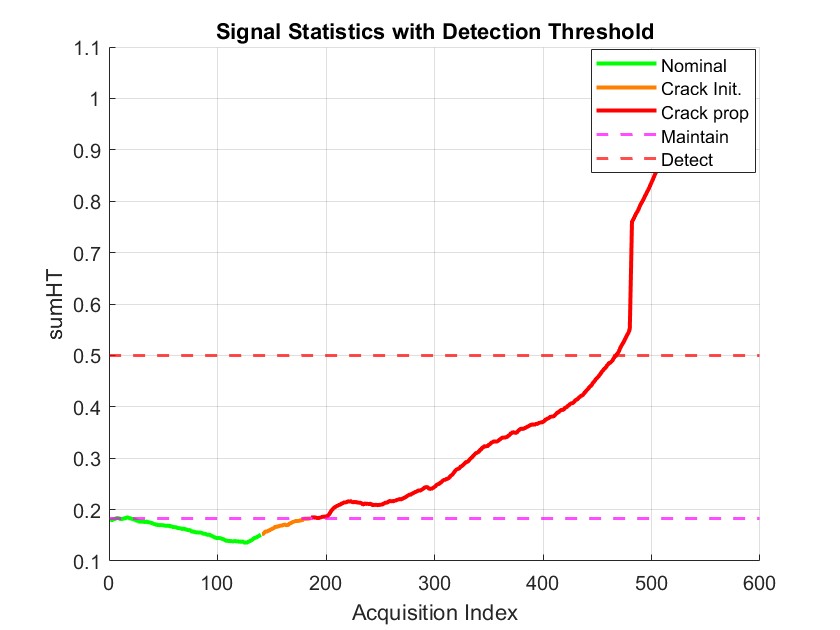}
    \vspace{-.1cm}
	\caption{The progression from a nominal (healthy) state through initial crack initiation to the propagation stage as captured by the Sum Hunting Tooth (sumHT) analysis.}
	\label{fig:sumHT}
\end{figure}

\begin{table}[h]
\centering
\caption{Comparison of PLANET TSA, HT TSA, and LDME Performance (cycles)}
\label{tab:tsa_comparison}
\setlength{\tabcolsep}{1pt}    
\renewcommand{\arraystretch}{0.9}
\footnotesize                  
\begin{tabular}{lcccc}
\toprule
 & PLANET TSA & HT TSA & LDME & \% Improvement  \\
\midrule
Detection cycle ($c_d$) & 355 & 284 & 198 & 30\% \\
Maintenance cycle ($c_m$) & 465 & 383 & 365 & 5\% \\
\bottomrule
\end{tabular}
\normalsize
\end{table}

These results show LDME provides earlier warnings under the stated protocol. To avoid overstating the claims we emphasize the following caveats: (i) the timing numbers are specific to the HUMS 2023 benchmark and our detector calibration; (ii) ROC/PR curves depend on threshold selection.  We report full ROC/PR curves in the supplementary figures and list area-under-curve values there.

\section{Discussion and critical assessment}
\label{sec:discussion}
We now place LDME in context and provide a candid that highlights both strengths and where further work is needed.

\subsection{ Effectiveness Analysis of the LDME method}
To assess the necessity and effectiveness of each component, ablation experiments were conducted by systematically removing individual stages of the LDME framework while keeping all remaining operators unchanged. Let $\mathcal{F} = \mathcal{R} \circ \mathcal{E} \circ \mathcal{D}$ denote the complete mapping from the raw signal $Y$ to the reconstructed output, where $\mathcal{D}$ represents dual-path denoising, $\mathcal{E}$ the multiscale decomposition and enhancement (DWT followed by fractional filtering), and $\mathcal{R}$ the reconstruction and condition-oriented spectral fusion. Removing $\mathcal{D}$ increases the variance of high-frequency components, leading to unstable wavelet coefficients and reduced sparsity of fault-related modes. Suppressing the multiscale operator within $\mathcal{E}$ collapses scale separation, causing mode mixing and spectral leakage, while disabling fractional filtering reduces the persistence of fault-induced modulations by eliminating long-memory effects intrinsic to damaged mechanical systems. Finally, omitting the spectral fusion stage degrades fault energy concentration at characteristic orders without improving noise suppression. In all cases, the ablated operators produce a monotonic decrease in fault separability and energy concentration, confirming that no individual component is redundant and that the effectiveness of the LDME framework arises from the coupled action of its operators rather than from any single stage in isolation.

\subsection{Strengths}
\begin{itemize}
    \item \textbf{Practical early-warning performance:} On the HUMS benchmark and simulated experiments, LDME consistently signals earlier than HT-TSA under the chosen calibration.
    \item \textbf{Interpretability:} Each processing stage has a clear physical or mathematical meaning (wavelet bands, SG shape preservation, TKEO energy amplification, fractional memory enhancement, harmonic indicators). This is important for safety-critical applications.
    \item \textbf{Reproducibility:} We supply an explicit protocol and implementation details to enable independent verification.
\end{itemize}

\subsection{Limitations and novelty discussion}
Even though LDME addresses an important and practical problem, a frank assessment reveals the method's current position in the research spectrum. The LDME pipeline is a thoughtful aggregation of well-established components: DWT, Savitzky--Golay smoothing, least-squares fitting, TKEO and fractional calculus each have long histories in signal processing. The present work does not yet introduce a fundamentally new signal processing algorithm with theoretical convergence guarantees or provable optimality. Instead, LDME is best characterized as a practically effective, physically motivated pipeline.

For Transactions-level novelty the manuscript would benefit from one (or both) of the following:
\begin{enumerate}
    \item \textbf{A unified theoretical framework.} For example, deriving LDME as the discretization or approximation of a single fractional-variational functional that admits a unique minimizer and for which TKEO and Hadamard--Caputo filtering arise naturally as Euler--Lagrange conditions; or
    \item \textbf{Rigorous statistical analysis.} Bounds on detection delay and false-alarm rate as a function of SNR, fractional order $\beta$, and model mismatch would transform an engineering pipeline into a provable detection algorithm.
\end{enumerate}

Until such theoretical consolidation is provided, the contribution remains strongly engineering-oriented. That said, this is not a fatal criticism: many field-changing tools started as pipelines and later received theoretical foundations. We therefore present LDME both as a useful engineering solution and as a candidate for deeper theoretical work.

\subsection{Practical recommendations}
Based on our experiments and analysis we recommend practitioners:
\begin{itemize}
    \item calibrate the fractional order $\beta$ using cross-validation on healthy segments and synthetic faults that match expected physics;
    \item use conservative thresholds for initial deployment to keep false alarms low and then adapt thresholds with operational feedback; and
    \item combine LDME with physics-based models (e.g., finite-element predictions of crack growth) for prognostic tasks.
\end{itemize}

\section{Conclusions and future work}
We introduced LDME, a layered, interpretable pipeline that integrates fractional-domain and classical methods for early gear-fault identification. Under the disclosed methodologies, evaluation on the Case Western Reserve University (CWRU) bearing dataset \cite{CWRU}, controlled simulations, and the HUMS 2023 benchmark shows better early-warning performance compared to representative TSA baselines. We state clearly that LDME is not yet a provably innovative signal-processing operator, but rather an engineering synthesis. Our main focus for future study is to close that gap using statistical detectability analysis or fractional-variational derivations. Multi-sensor fusion, online/real-time implementation, and transfer-learning methods that maintain interpretability while using benchmark datasets like CWRU for validation are further directions.

\section{Reproducibility and code availability}
All code implementing LDME, the simulation scripts, and experiment configuration files (including parameter values for fractional order, SG window, wavelet family and detector thresholds) are available in the project repository (link provided in the supplementary material). The HUMS 2023 benchmark links and acquisition instructions are included so that readers can reproduce the HUMS experiments exactly.

\bibliography{mybibliography}

\end{document}